\documentclass[11pt]{article}
\usepackage{fullpage}
\usepackage{amssymb,latexsym,amsmath}     
\usepackage [latin1]{inputenc}
\usepackage{graphicx}
\usepackage{dsfont}
\usepackage{amsthm}
\usepackage{color}
\usepackage{float,todonotes}
\usepackage{amsmath}
\usepackage{comment}
\usepackage{tcolorbox}
\usepackage[colorlinks=true,linkcolor=blue, citecolor=blue,urlcolor=blue]{hyperref}
  
\usepackage{algorithm}
\usepackage{algorithmic}

\theoremstyle{definition}
\newtheorem{definition}{Definition}[section]
\theoremstyle{plain}
\newtheorem{theorem}{Theorem}[section]
\newtheorem{lemma}{Lemma}[section]
\newtheorem{claim}{Claim}[section]
\newtheorem{observation}{Observation}

\theoremstyle{remark}

\newtheorem*{observation*}{Observation}

\usepackage{xcolor}
\usepackage{tikz}
\usetikzlibrary{arrows.meta,calc,positioning}

\definecolor{fcBlue}{HTML}{2563A6}
\definecolor{fcInk}{HTML}{243247}
\definecolor{fcMuted}{HTML}{768494}
\definecolor{fcRule}{HTML}{CFD8E3}

\tikzset{
  fc base/.style={
    font=\small,
    line cap=round,
    line join=round
  },
  fc heading/.style={
    font=\sffamily\bfseries\small,
    text=fcInk
  },
  fc note/.style={
    font=\sffamily\scriptsize,
    text=fcMuted
  },
  fc arrow/.style={
    -{Latex[length=2mm,width=1.3mm]},
    draw=fcBlue,
    line width=0.65pt
  },
  fc dimension/.style={
    {Latex[length=1.5mm]}-{Latex[length=1.5mm]},
    draw=fcMuted,
    line width=0.55pt
  }
}

\begin{document}
\title{Subexponential Parameterized Approximation for Group-Fair Euclidean \texorpdfstring{$k$}{k}-Center}

\author {
    Rudranarayan Kar\textsuperscript{\rm 1,2},
    Diya Roy\textsuperscript{\rm 1,2},
    Abhishek Sahu\textsuperscript{\rm 1,2}\\
    \small\textsuperscript{1} School of Computer Sciences, National Institute of Science Education and Research, Bhubaneswar\\
    \small\textsuperscript{2} Homi Bhabha National Institute (HBNI), Training School Complex, Anushakti Nagar, Mumbai, India
}
\date{}

\maketitle
\begin{abstract}
 Clustering is a basic tool for grouping data into meaningful classes, but in many applications the resulting clusters must also respect the fairness conditions among different groups. We study this direction for the group-fair Euclidean $k$-center problem in fixed dimension, where the input points are partitioned into disjoint color classes, and every cluster must satisfy given lower and upper bounds on the fraction of each color. 

 We give a separator-based recursive approximation algorithm for the continuous-center setting. The main idea is to first reduce the continuous search space to a finite candidate center set and then use the geometric sphere separators to break the instance into smaller subproblems. 

 Our algorithm gives a bicriteria approximation: the clustering radius is within a factor of $(1+\epsilon)$ of the optimum exactly fair solution, while the fairness bounds are relaxed only by a multiplicative factor. For fixed dimension, number of colors, and approximation parameters, the running time is polynomial in input size and subexponential in $k$. A recent IJCAI-ECAI2026 result gives a $(1+\epsilon)$-approximation for group-fair $k$-center in bounded doubling metrics with running time $(k/\epsilon)^{\mathcal{O}(k)}n^{\mathcal{O}(1)}$. In the fixed-dimensional Euclidean setting considered here, we improve the dependence on $k$ to subexponential, while allowing continuous centers. The fairness guarantees of the two results are different: the previous work allows an additive fairness violation, whereas our algorithm uses a controlled multiplicative relaxation.
\end{abstract}

\section{Introduction}
Clustering is a basic tool for dividing a dataset into meaningful groups. Among the classical clustering objectives, the $k$-center problem aims to select at most $k$ centers and assign every point to one of them so that the maximum assignment distance is minimized. In general metrics, the problem is NP-hard and admits a polynomial-time $2$-approximation, which is optimal unless $\mathrm{P}=\mathrm{NP}$ \cite{DBLP:journals/tcs/Gonzalez85,DBLP:journals/mor/HochbaumS85}. Computing an optimal solution is also hard when parameterized by the number $k$ of centers \cite{DBLP:journals/talg/DemaineFHT05}. Many variants of $k$-center have therefore been studied, including clustering with outliers \cite{DBLP:conf/soda/CharikarKMN01}, anonymity constraints \cite{DBLP:journals/talg/AggarwalPFTKKZ10}, capacities \cite{DBLP:conf/focs/CyganHK12}, lower bounds \cite{DBLP:journals/tcs/LiangWXF25}, and fairness constraints.

A clustering obtained only by minimizing distances may represent some demographic groups poorly. For example, a cluster may contain very few points from a minority group even when that group forms a significant part of the complete dataset. This has led to several notions of fairness in clustering, including group fairness \cite{DBLP:conf/nips/Chierichetti0LV17,DBLP:conf/nips/BeraCFN19,DBLP:conf/approx/Bercea0KKRS019}, proportional fairness \cite{DBLP:conf/icml/ChenFLM19,DBLP:conf/icalp/Micha020}, and individual fairness \cite{DBLP:conf/icml/MahabadiV20,DBLP:conf/nips/NegahbaniC21}. In this paper, we focus on group fairness. Under group fairness, every point belongs to a demographic group, represented by a color. For every color $i$, lower and upper bounds $\alpha_i$ and $\beta_i$ are given, and the proportion of color $i$ in every nonempty cluster must lie between these bounds. Thus, fairness depends not only on the selected centers but also on the assignment of points to them. In particular, assigning every point to its nearest center may fail to produce a fair solution.

Chierichetti et al.\cite{DBLP:conf/nips/Chierichetti0LV17} initiated the study of fair clustering for two colors and introduced the fairlet framework. A fairlet is a small group of points that already satisfies the required balance condition. Their approach first partitions the input into fairlets and then applies a standard clustering algorithm to these fairlets. The group-fairness model was later generalized to multiple colors and lower and upper proportion bounds. Bercea et al. \cite{DBLP:conf/approx/Bercea0KKRS019} obtained a $5$-approximation for exact ratio preservation and a $3$-approximation under essential fairness, where the number of points of each color may differ by at most one from a fair fractional assignment. Ahmadian et al. \cite{DBLP:conf/kdd/AhmadianE0M19} considered upper-bound representation constraints and gave a $3$-approximation with a bounded additive fairness violation. Bera et al.\cite{DBLP:conf/nips/BeraCFN19} developed a more general framework that also allows protected groups to overlap. Let $\Delta$ be the maximum number of groups to which any point belongs. For fair $k$-center, their framework gives a $4$-approximation with an additive fairness violation of at most $4\Delta+3$. Harb and Lam \cite{DBLP:conf/nips/HarbL20} subsequently obtained a scalable $3$-approximation while retaining the $4\Delta+3$ violation for overlapping groups. For disjoint groups, smaller additive violations are possible. These results show that overlapping groups make the assignment problem substantially more difficult. The present paper considers disjoint color classes.

Parameterized algorithms provide another way to obtain stronger approximation guarantees. Goyal and Jaiswal \cite{DBLP:journals/tcs/GoyalJ23} gave a $2$-approximation for constrained $k$-center, including group-fair $k$-center, in FPT time parameterized by the number of centers and colors. They also showed that an approximation factor below $2$ cannot be obtained in FPT time for general metrics under standard parameterized-complexity assumptions. Fair coresets have also been studied for the $k$-median and $k$-means objectives \cite{DBLP:conf/icalp/BandyapadhyayFS21}. Moreover, even when the centers are fixed, finding a fair assignment can remain computationally difficult \cite{DBLP:conf/nips/EsmaeiliBSD21}. This explains why the assignment step cannot be treated as a simple nearest-center assignment. Several works have considered fair $k$-center in structured, distributed, and low-dimensional metric spaces. Ceccarello et al. \cite{DBLP:conf/www/CeccarelloPP24} studied fast fair $k$-center algorithms in doubling metrics, while Wu et al. \cite{DBLP:conf/atal/WuFH0024} considered the problem in a distributed setting. In fixed-dimensional Euclidean space, Moritaka et al. \cite{DBLP:journals/corr/abs-2511-07660} studied fair $k$-center in the plane with absolute lower and upper bounds on the number of points of each color. Their fairness requirements differ from the proportional bounds considered here.

The result closest to ours is the recent work of Wu et al. \cite{ijcai2026p92} on fair clustering in bounded doubling metrics. For group-fair $k$-center, they use the same lower- and upper-bound proportion model considered in this paper and obtain a $(1+\epsilon)$-approximation for the radius with an additive fairness violation of one. Their algorithm runs in $\left(\frac{k}{\epsilon}\right)^{\mathcal{O}(k)}n^{\mathcal{O}(1)}$ time. Their paper also studies data-summarization fairness, whereas we focus only on group fairness. Moreover, their centers are selected from the input points, while our centers may be placed anywhere in the Euclidean space.

Subexponential parameterized algorithms are already known for the ordinary Euclidean $k$-center problem. In particular, geometric separator methods lead to running times whose exponents involve $k^{1-1/d}$ in dimension $d$ \cite{DBLP:journals/algorithmica/AgarwalP02}. Closely related lower bounds are known for the discrete Euclidean problem \cite{DBLP:conf/compgeom/ChitnisS22}. These results motivate the question of whether geometric separators can also provide a subexponential algorithm when group-fairness constraints are present. The main difficulty is that a cluster crossing a separator may receive points of different colors from both sides, and its fairness cannot be checked independently within either side.
\subsection{Our Contributions}
In this paper, we study the group-fair $k$-center problem in fixed-dimensional Euclidean space with continuous centers. Our main contribution is a separator-based structural approach that uses the geometry of Euclidean space to reduce the search space and divide the instance into smaller subproblems while keeping only a small number of active centers.

Based on this approach, we obtain a bicriteria approximation algorithm for group-fair Euclidean $k$-center. Given an exactly fair solution, our algorithm returns a clustering whose radius is within a factor $(1+\epsilon)$ of the optimum exactly fair radius, while allowing a controlled multiplicative relaxation of the fairness constraints. For fixed dimension, number of colors, and approximation parameters, the algorithm runs in polynomial time in the input size with subexponential dependence on $k$.

\begin{theorem}
For any real numbers $0<\varepsilon\leq 1$ and $\lambda>1$, there exists a bicriteria approximation algorithm for the group-fair Euclidean $k$-center problem with continuous centers in fixed dimension $d$. The algorithm achieves a $(1+\epsilon)$-approximation for the clustering radius and a multiplicative $\lambda$-relaxation of the fairness constraints. For fixed number of colors $c$, dimension $d$, and $\lambda$, its running time is $f(k,\varepsilon)\,n^{\mathcal{O}(d)},\text {   where  }f(k,\epsilon)=\left(\frac{k}{\epsilon}\right)^{
\mathcal{O}_{c,d,\lambda}\left(k^{1-1/d}/\epsilon\right)}.$

In particular, for fixed $\epsilon$, the dependence on $k$ is subexponential.
\end{theorem}
 \noindent\textbf{Organization of the Paper.}  In Section \ref{pre}, we introduce the notation, preliminary concepts, and the fairness model used throughout the paper. In Section \ref{over}, we give an overview of the main technical ideas behind our separator-based approach. In Section \ref{algo} , we present the algorithm and describe its main recursive steps. In Section \ref{correct}, we prove the correctness of the algorithm and establish our main bicriteria approximation result. In Section \ref{sec:runtime}, we analyze the running time and show the subexponential dependence on $k$. Finally, in Section \ref{sec:conclusion}, we conclude the paper and discuss possible directions for further study.

\section{Preliminaries}\label{pre}

For any integer $m\ge 1$, let $[m]=\{1,\ldots,m\}$. Throughout the paper, we work in the Euclidean space $\mathbb{R}^d$, where $d\ge 2$ is fixed, and let $P\subseteq\mathbb{R}^d$ denote the set of input points. For any two points $u,v\in\mathbb{R}^d$, we write $d(u,v)=\|u-v\|_2$ for their Euclidean distance. For a point $v\in\mathbb{R}^d$ and a radius $r\ge 0$, let $B(v,r)=\{x\in\mathbb{R}^d : d(x,v)\le r\}$ denote the Euclidean ball of radius $r$ centered at $v$. For a point $v\in\mathbb{R}^d$ and a set of centers $S\subseteq\mathbb{R}^d$, we use $d(v, S)=\min_{s\in S} d(v,s)$ to denote the distance from $v$ to its nearest center in $S$. For a finite set $X\subseteq\mathbb{R}^d$, let
$\operatorname{diam}(X)=\max_{x,y\in X} d(x,y).$  Unlike the discrete $k$-center setting, the centers in our problem are continuous and may be chosen from any point of $\mathbb{R}^d$.

We use standard graph notation. For an undirected graph $G=(V, E)$, $V(G)$ and $E(G)$ denote its vertex and edge sets, respectively. A connected component is a maximal connected subgraph of $G$. The notation $\mathcal{O}_d(\cdot)$ means that the hidden constant may depend on $d$. Similarly, $\mathcal{O}_{c,d}(\cdot)$ allows dependence on $c$ and $d$.

We first define the Euclidean $k$-center problem. Let $P\subseteq \mathbb{R}^d$ be a set of $n$ points and let $k$ be a positive integer. In our setting, the centers are continuous and can be chosen from anywhere in $\mathbb{R}^d$.

\begin{definition}[Euclidean $k$-Center]
An instance of the Euclidean $k$-center problem is given by $(P,k)$. The goal is to find a set $C\subseteq\mathbb{R}^d$, with $|C|\le k$, such that $\max_{p\in P} d(p,C)$ is minimized, where $d(p,C)=\min_{a\in C}\|p-a\|_2.$
\end{definition}

We now introduce the group-fairness condition used in this paper. Suppose the input points are partitioned into $c$ disjoint color classes $\mathcal{P}=\{P_1,\ldots,P_c\},$
such that $P_1\cup P_2\cup\cdots\cup P_c=P$ and $P_i \cap P_j =\emptyset \text{ for any }i \neq j.$
Every point in $P_i$ is said to have color $i$. For every color $i\in[c]$, we are given two values $\alpha_i$ and $\beta_i$, where $0\le \alpha_i\le \beta_i\le 1.$
These values specify the lower and upper proportion of color $i$ allowed in every cluster.

\begin{definition}[Group-Fair Euclidean $k$-Center]
An instance of the group-fair Euclidean $k$-center problem, denoted by GF-$k$-CEN, is given by $(P,\mathcal{P},\alpha, \beta,k),$
where $\alpha=(\alpha_1,\ldots,\alpha_c), \beta=(\beta_1,\ldots,\beta_c).$ The goal is to find a set $C\subseteq\mathbb{R}^d$ of at most $k$ centers and an assignment $\phi:P\rightarrow C$ that minimizes $\max_{p\in P}\|p-\phi(p)\|_2,$ while satisfying, for every nonempty cluster with center $a\in C$ and every color $i\in[c]$, $$\alpha_i \le \frac{|\{p\in P_i:\phi(p)=a\}|} {|\{p\in P:\phi(p)=a\}|} \le \beta_i.$$

\noindent A pair $(C,\phi)$ satisfying these conditions is called an exactly $(\alpha,\beta)$-fair feasible solution.
\end{definition}

For our approximation algorithm, we also use the following relaxed fairness notion.

\begin{definition}[$\lambda$-Relaxed Group Fairness]
For $\lambda>1$, a clustering is called $\lambda$-relaxed
$(\alpha,\beta)$-fair if, for every nonempty cluster and every color $i\in[c]$,
$$\frac{\alpha_i}{\lambda} \le \frac{|\{p\in P_i:\phi(p)=a\}|} {|\{p\in P:\phi(p)=a\}|} \le\min\{1,\lambda\beta_i\}.$$
\end{definition}

Let $R$ denote the minimum radius of an exactly
$(\alpha,\beta)$-fair solution using at most $k$ continuous centers. Our algorithm returns a clustering of radius at most $(1+\epsilon)R$ while satisfying the above $\lambda$-relaxed fairness condition.

\medskip
\noindent\textbf{Parameterized Complexity.}
A parameterized problem has an input of size $n$ together with a parameter $k$. An algorithm is called fixed-parameter tractable, or FPT, with respect to $k$ if its running time can be written as $f(k)n^{\mathcal{O}(1)},$ where $f$ is an arbitrary computable function depending only on the parameter. More generally, when other quantities such as $1/\varepsilon$ are treated as parameters, the running time may be written as $f(k,1/\epsilon)n^{\mathcal{O}(1)}.$
In contrast, an algorithm with running time $n^{f(k)}$
is said to be in XP with respect to $k$.

We call the dependence on $k$ subexponential if it is of the form $2^{o(k)}.$ For example, for every fixed $d\ge 2$, $2^{\mathcal{O}(k^{1-1/d}\log k)}$ is subexponential in $k$. An FPT approximation scheme gives, for every $\varepsilon>0$, a $(1+\varepsilon)$-approximation with an FPT running time in the relevant parameters. The geometric running-time bounds are stated in the standard real-arithmetic model for fixed-dimensional Euclidean algorithms.

\medskip
{\noindent\textbf{Basic Properties.}} The following three properties will be useful throughout the paper. They capture the basic facts about fairness, count distortion, and minimum enclosing balls that are needed in our algorithm.

\medskip
\noindent\textbf{Property 1.} \label{unifair} If non empty clusters $X_1,X_2,\dots,X_t$ are each $(\alpha,\beta)$-fair, then their disjoint union ${X=\bigcup_{j=1}^{t} X_j}$ is also $(\alpha,\beta)$-fair.

\medskip
\noindent\textbf{Property 2.} \label{sqmiu} Let $x_1,x_2,\dots,x_c \geq 0$ and $x_1',x_2',\dots,x_c' \geq 0$ with totals $N=\sum_i x_i>0$ and $N'=\sum_i x_i'>0$. Suppose that for some $\mu \geq 1$, $\frac{x_i}{\mu} \leq x_i' \leq \mu x_i \text{ holds for every }i\in[c].$ Then, $$ \frac{x_i}{\mu^2N} \leq \frac{x_i'}{N'} \leq \mu^2 \frac{x_i}{N}$$

\noindent\textbf{Property 3.} \label{minpoints} For every finite nonempty set of points $X \subseteq \mathbb{R}^d$, the radius of minimum enclosing ball of $X$ can be always determined by at most $d+1$ many points \cite{DBLP:conf/birthday/Welzl91}.

For Property~\ref{unifair}, each color proportion in the union is a weighted average of the proportions in the individual clusters. For Property~\ref{sqmiu}, summing the count inequalities gives $N/\mu\leq N'\leq\mu N$, and dividing the bounds for $x_i'$ by these total-size bounds gives the stated inequalities.

\section{Technical Overview} \label{over}

Our goal is to form at most $k$ clusters whose points lie close to their centers, while keeping the proportion of each color within its prescribed range. We allow a small increase in radius and a controlled relaxation of fairness. The main idea is to reduce the possible center locations to a finite set, divide the problem into smaller regions, and keep enough information to combine their assignments.
\medskip

\noindent\textbf{Preprocessing Phase.} 
The algorithm first guesses the radius $R$. In $d$-dimensional space, the minimum enclosing ball of any cluster is determined by at most $d+1$ points. Hence, after moving every center to the minimum enclosing ball of its cluster, the fairness does not change, and the radius does not increase. Therefore, it is enough to consider the radii obtained from all subsets of at most $d+1$ input points.
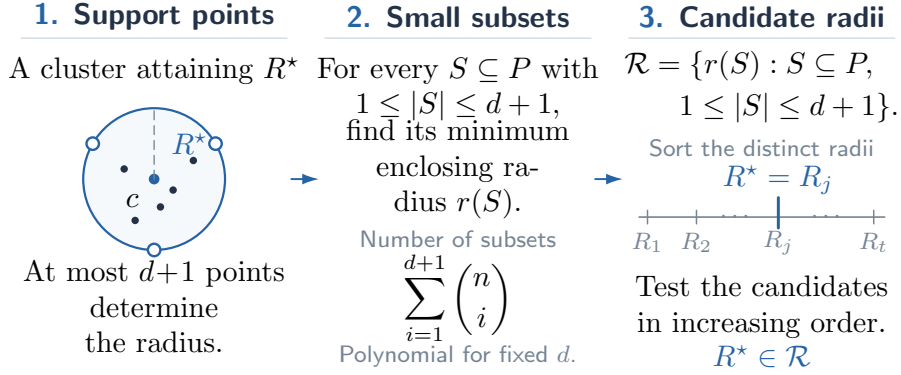
\begin{figure}[htbp]
\centering
\resizebox{0.7\textwidth}{!}{%
\begin{tikzpicture}[fc base]

\path[use as bounding box]
  (0,-0.10) rectangle (10.35,4.50);

\node[fc heading] at (1.575,4.15)
  {\textcolor{fcBlue}{1.} Support points};
\node[fc heading] at (5.175,4.15)
  {\textcolor{fcBlue}{2.} Small subsets};
\node[fc heading] at (8.775,4.15)
  {\textcolor{fcBlue}{3.} Candidate radii};

\foreach \x in {0,3.60,7.20}
  \draw[fcRule,line width=0.45pt]
    (\x,3.90)--++(3.15,0);

\node at (1.575,3.52)
  {A cluster attaining $R^\star$};

\coordinate (C) at (1.575,2.23);

\draw[fcBlue,fill=fcBlue!4,line width=0.75pt]
  (C) circle[radius=0.84cm];

\foreach \x/\y in {
  1.22/2.35,1.65/1.90,1.80/2.10,
  1.35/1.74,2.04/2.45}
  \fill[fcInk] (\x,\y) circle[radius=1.25pt];

\foreach \angle in {30,150,270}
  \path (C) ++(\angle:0.84cm)
    node[
      circle,draw=fcBlue,fill=white,
      line width=0.8pt,
      minimum size=4.2pt,inner sep=0pt
    ] {};

\fill[fcBlue] (C) circle[radius=1.8pt];
\node[below left=2pt] at (C) {$c$};

\draw[fcMuted,densely dashed,line width=0.45pt]
  (C)--++(0,0.84)
  node[midway,right=2pt,text=fcBlue] {$R^\star$};

\node[align=center,text width=3.10cm] at (1.575,0.72)
  {At most $d+1$ points\\ determine the radius.};

\draw[fc arrow] (3.20,2.23)--(3.53,2.23);
\draw[fc arrow] (6.80,2.23)--(7.13,2.23);

\node at (5.175,3.52)
  {For every $S\subseteq P$ with};

\node at (5.175,3.10)
  {$1\le |S|\le d+1$,};

\node[align=center,text width=3.10cm] at (5.175,2.35)
  {find its minimum\\ enclosing radius $r(S)$.};

\node[fc note] at (5.175,1.52)
  {Number of subsets};

\node at (5.175,0.82)
  {$\displaystyle\sum_{i=1}^{d+1}\binom{n}{i}$};

\node[fc note] at (5.175,0.13)
  {Polynomial for fixed $d$.};

\node at (8.775,3.32) {$\begin{aligned}
  \mathcal{R}&=\{r(S):S\subseteq P,\\[-1pt]
             &\quad 1\le |S|\le d+1\}.
\end{aligned}$};

\node[fc note] at (8.775,2.60)
  {Sort the distinct radii};

\draw[fcMuted,line width=0.55pt]
  (7.32,1.78)--(10.23,1.78);

\foreach \x/\lab in {
  7.42/1,8.00/2,8.97/j,10.10/t}{
  \draw[fcMuted,line width=0.5pt]
    (\x,1.70)--(\x,1.86);
  \node[fc note] at (\x,1.48) {$R_{\lab}$};
}

\node[fc note] at (8.48,1.78) {$\cdots$};
\node[fc note] at (9.56,1.78) {$\cdots$};

\draw[fcBlue,line width=1pt]
  (8.97,1.65)--(8.97,2.01);

\node[text=fcBlue] at (8.97,2.22)
  {$R^\star=R_j$};

\node[align=center,text width=3.10cm] at (8.775,0.72)
  {Test the candidates\\ in increasing order.};

\node[text=fcBlue] at (8.775,0.13)
  {$R^\star\in\mathcal{R}$};

\end{tikzpicture}%
}
\caption{In $\mathbb{R}^{d}$, a minimum enclosing ball is determined
by at most $d+1$ points. Enumerating these small subsets therefore
gives a finite candidate set containing the optimal radius $R^\star$.}
\label{fig:finite-radius-guess}
\end{figure}

For a fixed radius $R$, we build a graph on $P$ by connecting two points whose distance is at most $2R$. Any two points assigned to the same radius-$R$ center must belong to the same connected component. Hence, different connected components can be considered separately. Moreover, if a feasible solution using at most $k$ centers exists, there are at most $k$ such components, and each component has diameter $\mathcal{O}(kR)$. Therefore, each component lies in a bounded region. We enlarge this region slightly and place a fine grid inside it. The grid points are then used as the candidate centers in the remaining algorithm.

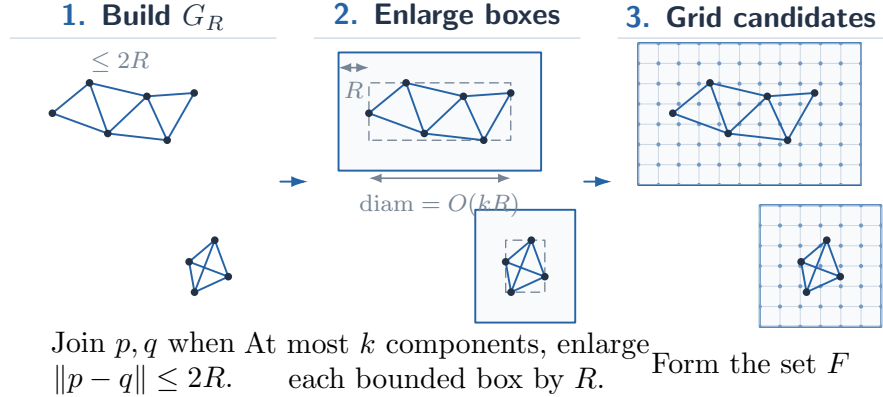
\begin{figure}[htbp]
\centering
\resizebox{0.7\textwidth}{!}{%
\begin{tikzpicture}[fc base]

\def\fcComponentA{%
  \draw[fcBlue,line width=0.65pt]
    (0,0.2)--(0.55,0.65)--(0.82,-0.1)--cycle
    (0.55,0.65)--(1.4,0.45)--(2.1,0.5)--(1.7,-0.2)
    --(0.82,-0.1)--(1.4,0.45)--(1.7,-0.2);
  \foreach \x/\y in {
    0/0.2,0.55/0.65,0.82/-0.1,
    1.4/0.45,1.7/-0.2,2.1/0.5}
    \fill[fcInk] (\x,\y) circle[radius=1.6pt];
}

\def\fcComponentB{%
  \draw[fcBlue,line width=0.65pt]
    (0,0)--(0.38,0.32)--(0.58,-0.22)--(0.08,-0.45)--cycle
    (0,0)--(0.58,-0.22)
    (0.38,0.32)--(0.08,-0.45);
  \foreach \x/\y in {
    0/0,0.38/0.32,0.58/-0.22,0.08/-0.45}
    \fill[fcInk] (\x,\y) circle[radius=1.6pt];
}

\path[use as bounding box]
  (0,-0.35) rectangle (10.35,4.50);

\node[fc heading] at (1.575,4.15)
  {\textcolor{fcBlue}{1.} Build $G_R$};
\node[fc heading] at (5.175,4.15)
  {\textcolor{fcBlue}{2.} Enlarge boxes};
\node[fc heading] at (8.775,4.15)
  {\textcolor{fcBlue}{3.} Grid candidates};

\foreach \x in {0,3.60,7.20}
  \draw[fcRule,line width=0.45pt]
    (\x,3.90)--++(3.15,0);

\draw[fc arrow] (3.20,2.20)--(3.53,2.20);
\draw[fc arrow] (6.80,2.20)--(7.13,2.20);

\begin{scope}[shift={(0.50,2.85)},scale=0.8]
  \fcComponentA
\end{scope}

\begin{scope}[shift={(2.12,1.25)},scale=0.8]
  \fcComponentB
\end{scope}

\node[fc note] at (1.30,3.62) {$\le 2R$};

\begin{scope}[shift={(4.25,2.85)},scale=0.8]
  \draw[fcBlue,fill=fcBlue!3,line width=0.65pt]
    (-0.45,-0.65) rectangle (2.55,1.10);
  \draw[fcMuted,densely dashed,line width=0.45pt]
    (0,-0.20) rectangle (2.10,0.65);
  \fcComponentA
\end{scope}

\begin{scope}[shift={(5.87,1.25)},scale=0.8]
  \draw[fcBlue,fill=fcBlue!3,line width=0.65pt]
    (-0.45,-0.90) rectangle (1.03,0.77);
  \draw[fcMuted,densely dashed,line width=0.45pt]
    (0,-0.45) rectangle (0.58,0.32);
  \fcComponentB
\end{scope}

\draw[fc dimension] (3.89,3.55)--(4.25,3.55)
  node[midway,below=1pt,fc note] {$R$};

\draw[fc dimension] (4.25,2.24)--(5.93,2.24)
  node[midway,below=1pt,fc note]
    {$\operatorname{diam}=O(kR)$};

\foreach \xa/\ya/\xb/\yb in {
  7.44/2.16/10.08/3.84,
  8.88/0.48/10.32/1.92}{
  \draw[fcBlue,fill=fcBlue!3,line width=0.65pt]
    (\xa,\ya) rectangle (\xb,\yb);
  \begin{scope}
    \clip (\xa,\ya) rectangle (\xb,\yb);
    \draw[fcRule,line width=0.3pt,step=0.24]
      (7.20,0) grid (10.56,4.08);
    \foreach \ix in {1,...,13}
      \foreach \iy in {2,...,16}
        \fill[fcBlue!65]
          ({7.20+0.24*\ix},{0.24*\iy})
          circle[radius=0.7pt];
  \end{scope}
}

\begin{scope}[shift={(7.85,2.85)},scale=0.8]
  \fcComponentA
\end{scope}

\begin{scope}[shift={(9.37,1.25)},scale=0.8]
  \fcComponentB
\end{scope}

\node[align=center] at (1.575,0.05)
  {Join $p,q$ when\\ $\|p-q\|\le 2R$.};

\node[align=center] at (5.175,0.05)
  {At most $k$ components, enlarge \\ each bounded box by $R$.};

\node[align=center] at (8.775,0.05)
  { Form the set $F$};

\end{tikzpicture}%
}
\caption{The graph $G_R$ separates the input into connected components.
Each component lies in a bounded region, which is enlarged and covered
by a fine grid. The grid points are used as candidate centers.}
\label{fig:component-grid}
\end{figure}

\noindent\textbf {Gridding Phase.} Moving a center to a nearby grid point increases the radius only slightly; a sufficiently fine grid keeps it within \((1+\varepsilon)R\). If several centers move to the same location, we merge their clusters (by Property~\ref{unifair}, this preserves fairness and uses no more centers). This uses another property: merging fair clusters preserves fairness. The proportion of any color after merging is a weighted average of its proportions in the original clusters, so it remains within the prescribed range. 
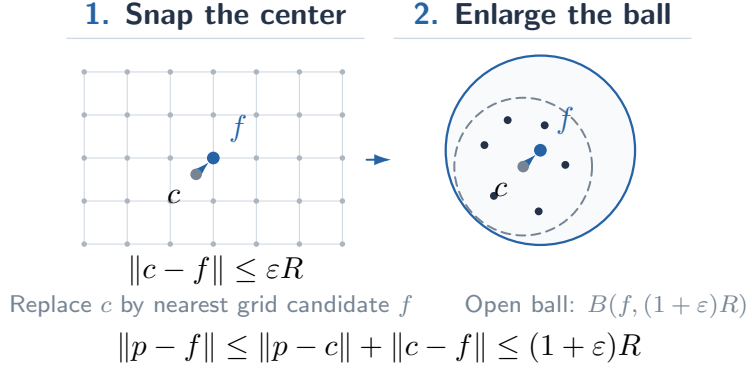
\begin{figure}[htbp]
\centering
\resizebox{0.5\textwidth}{!}{%
\begin{tikzpicture}[fc base]

\path[use as bounding box]
  (0,-0.15) rectangle (7.25,4.20);

\node[fc heading] at (1.725,3.95)
  {\textcolor{fcBlue}{1.} Snap the center};
\node[fc heading] at (5.525,3.95)
  {\textcolor{fcBlue}{2.} Enlarge the ball};

\draw[fcRule,line width=0.45pt]
  (0,3.69)--(3.45,3.69);
\draw[fcRule,line width=0.45pt]
  (3.80,3.69)--(7.25,3.69);

\foreach \x in {0.20,0.70,...,3.20}
  \draw[fcRule,line width=0.35pt]
    (\x,1.30)--(\x,3.30);

\foreach \y in {1.30,1.80,...,3.30}
  \draw[fcRule,line width=0.35pt]
    (0.20,\y)--(3.20,\y);

\foreach \x in {0.20,0.70,...,3.20}
  \foreach \y in {1.30,1.80,...,3.30}
    \fill[fcMuted!60] (\x,\y) circle[radius=0.9pt];

\coordinate (c) at (1.50,2.11);
\coordinate (f) at (1.70,2.30);

\draw[fc arrow,shorten <=2pt,shorten >=2pt]
  (c)--(f);

\fill[fcMuted] (c) circle[radius=1.9pt];
\fill[fcBlue] (f) circle[radius=2.1pt];

\node[below left=2pt] at (c) {$c$};
\node[above right=2pt,text=fcBlue] at (f) {$f$};

\node at (1.725,0.99)
  {$\|c-f\|\le\varepsilon R$};

\draw[fc arrow] (3.48,2.28)--(3.77,2.28);

\coordinate (C) at (5.30,2.20);
\coordinate (F) at (5.50,2.39);

\draw[fcBlue,fill=fcBlue!3,line width=0.75pt]
  (F) circle[radius=1.10cm];

\draw[fcMuted,densely dashed,line width=0.65pt]
  (C) circle[radius=0.80cm];

\foreach \x/\y in {
  4.85/2.45,5.12/2.74,5.55/2.68,
  5.79/2.22,5.48/1.68,4.96/1.86}
  \fill[fcInk] (\x,\y) circle[radius=1.3pt];

\draw[fc arrow,shorten <=2pt,shorten >=2pt]
  (C)--(F);

\fill[fcMuted] (C) circle[radius=1.9pt];
\fill[fcBlue] (F) circle[radius=2.1pt];

\node[below left=2pt] at (C) {$c$};
\node[above right=2pt,text=fcBlue] at (F) {$f$};

\node[fc note] at (3.625,0.58)
  {Replace $c$ by nearest grid candidate $f$\qquad
   Open ball: $B(f,(1+\varepsilon)R)$};

\node at (3.625,0.10)
  {$\|p-f\|\le\|p-c\|+\|c-f\|
    \le(1+\varepsilon)R$};

\end{tikzpicture}%
}
\caption{Effect of moving a center to a nearby grid point.
The center $c$ moves to its nearest grid point $f$ by a distance
of at most $\varepsilon R$. Therefore, every point previously
covered within radius $R$ remains covered within radius
$(1+\varepsilon)R$.}
\label{fig:grid-displacement}
\end{figure}

We next divide the problem using a separator sphere.
Each candidate center $f\in F$ defines a ball
$D_f=B(f,r)$, where
$r=R+\varepsilon R/20<(1+\varepsilon)R$.
This ball contains precisely those points that can
be assigned to $f$ without exceeding radius $r$.
 Since the candidate centers lie on a grid, only a bounded number of these balls can cover the same point. The separator theorem gives a sphere that places a fixed fraction of the selected centers on each side, while only a small number of their balls cross the sphere. We guess these crossing balls and solve the inside and outside parts recursively. 
\medskip

\noindent\textbf{Sphere-Separator Phase.} We next use a geometric separator to divide the problem into smaller parts. For every candidate center \(f\in F\), consider the ball \(D_f=B(f,r)\), which contains all points at distance at most \(r\) from \(f\). A sphere is the boundary of a ball; in two dimensions, it is simply a circle. We use such a sphere as a separator to divide the candidate balls into three groups: those lying completely inside, those lying completely outside, and those crossing the separator.

The main idea behind our subexponential algorithm comes from the sphere-separator theorem. Since the candidate centers are placed on a grid, the family of candidate balls has bounded ply \(p=\mathcal{O}_d(\varepsilon^{-d})\). Now consider any \(\ell\) selected candidate balls. The theorem guarantees the existence of a separator sphere \(\gamma\) such that at most \(\beta_d\ell\) balls lie completely on either side, where \(\beta_d=(d+1)/(d+2)<1\), and only $\mathcal{O}_d\left({\ell^{1-1/d}}/{\varepsilon}\right)$ balls cross the separator. In particular, for \(d=2\), only \(\mathcal{O}(\sqrt{\ell}/\varepsilon)\) balls cross it. Thus, after identifying these crossing balls, we can divide the remaining selected centers into two smaller parts and solve them recursively.

This idea is similar to separator-based dynamic programming on graphs of small treewidth. In such algorithms, a small boundary connects different parts of the graph, and the information stored at this boundary allows the parts to be solved independently. We follow a similar idea, where the crossing balls form the boundary between the two geometric regions. However, our problem is more involved because a center whose ball crosses the separator may receive points of different colors from both sides. Therefore, we must also keep track of the number of points of each color assigned to these centers.

\medskip

\noindent\textbf{Why a More Careful Recursion Is Needed.} Although the separator gives a small boundary, directly applying a standard dynamic programming approach is not sufficient for our required running time. At each step, there may be polynomially many possible separator choices, and the recursion has \(\mathcal{O}_d(\log k)\) levels. More importantly, storing every possible color count for the crossing centers would create \(n^{\mathcal{O}(c k^{1-1/d}/\varepsilon)}\) states, leading to an XP-type running time rather than an FPT algorithm. To overcome this difficulty, we use a specially designed recursive algorithm. Instead of considering infinitely many separator spheres, we enumerate a finite family of representative spheres that preserve the inside, outside, and crossing classification of every candidate ball. For each representative separator, we guess the selected crossing centers and recursively solve the two sides. We keep these centers active until their complete clusters are known and store their color counts using multiplicative buckets rather than exact values. This allows us to control the number of states and obtain the required subexponential dependence on \(k\). We explain these steps below.
\medskip

\noindent\textbf{Recursion Phase. } With $\ell$ centers remaining, at most $ \mathcal{O}_d(\ell^{1-1/d}/\epsilon)$ selected balls need to cross. We enumerate the relevant separators and guess the selected crossing balls before solving the two sides recursively. A crossing cluster may receive differently colored points from both sides, whose union is fair. We therefore keep it active, record the color counts contributed by each side, and check fairness only after all contributions have been collected. A direct use of this idea gives a $k^{\mathcal{O}(k)}$-type algorithm; obtaining a subexponential running time requires the more careful recursion described below. 
\medskip

\noindent \textbf{Bucketing Subroutine.} Storing every possible color count would make the number of states too large and lead to an XP-type algorithm. We therefore use bucketing: zero is stored separately, while nearby positive counts are grouped into the same geometric interval. Thus, the algorithm remembers only an approximate value of each count instead of its exact value, which greatly reduces the number of choices. For one center-color pair, we use $L=\mathcal{O}(1+\log(k+1)\log(n+1)/\log\lambda)$ buckets. Along a recursive path there are at most $H_d=\mathcal{O}_d(k^{1-1/d}/\varepsilon)$ active centers, so a table has at most $L^{cH_d}$ count states. The factor involving $\log n$ can be absorbed into $n^{\mathcal{O}(1)}$ together with a parameter-dependent factor, as shown in Section~\ref{sec:runtime}.

This approximation is justified by Property~\ref{sqmiu}. If every color count changes by at most a factor $\mu$, then the total cluster size also changes by at most $\mu$. Therefore, every color proportion changes by at most a factor $\mu^2$. Hence, using sufficiently fine buckets introduces only a controlled error in fairness. This bucketing step reduces the state space and helps obtain a subexponential FPT algorithm instead of an XP algorithm. Adding approximate counts does not increase the error factor. Error is added only when the counts are rounded again. Since the recursion has $\mathcal{O}(\log k)$ levels, we choose narrow buckets so that the total error remains within the allowed fairness relaxation.
\medskip

\noindent \textbf{Flowing Subroutine.}  
Finally, we use a lower and upper bounded feasible flow for each color to construct an assignment consistent with the chosen bucket intervals. Every point is assigned to an eligible center, and the number of points of each color assigned to a center remains within its stored interval. Since all capacities are integral, a feasible flow gives an integral assignment. Once a cluster is complete, we check its actual color counts against the relaxed fairness bounds. The small number of active centers and the compact representation of their color counts give the required subexponential dependence on $k$, for fixed dimension, number of colors, and approximation parameters.

\begin{figure}[htbp]
\centering
\resizebox{0.6\textwidth}{!}{%
\begin{tikzpicture}[
  fc base,
  client/.style={
    circle,draw=fcBlue,fill=fcBlue!4,
    line width=0.65pt,minimum size=0.46cm,
    inner sep=0pt
  },
  center/.style={
    rectangle,draw=fcBlue,fill=fcBlue!4,
    line width=0.65pt,minimum width=0.58cm,
    minimum height=0.42cm,inner sep=0pt
  },
  terminal/.style={
    circle,draw=fcInk,fill=fcInk,text=white,
    minimum size=0.48cm,inner sep=0pt
  },
  capacity/.style={
    fc note,fill=white,inner sep=1.5pt
  }
]

\path[use as bounding box]
  (0,-0.10) rectangle (8.65,4.20);

\node[fc heading] at (4.325,4.02)
  {Lower-bounded assignment for color $i$};

\draw[fcRule,line width=0.45pt]
  (0,3.78)--(8.65,3.78);

\node[fc note] at (0.35,3.55) {Source};
\node[fc note] at (2.00,3.55) {Color-$i$ clients};
\node[fc note] at (5.05,3.55) {Selected centers};
\node[fc note] at (8.30,3.55) {Sink};

\node[terminal] (s) at (0.35,2.20) {$s$};
\node[terminal] (t) at (8.30,2.20) {$t$};

\foreach \j/\y in {1/3.10,2/2.50,3/1.90,4/1.30}
  \node[client] (p\j) at (2.00,\y) {$p_{\j}$};

\foreach \j/\y in {1/2.95,2/2.20,3/1.45}
  \node[center] (a\j) at (5.05,\y) {$a_{\j}$};

\foreach \j in {1,...,4}
  \draw[fc arrow] (s)--(p\j);

\node[capacity] at (0.99,3.18) {$[1,1]$};

\foreach \p/\a in {
  1/1,1/2,2/1,2/3,3/2,3/3,4/2,4/3}
  \draw[fc arrow] (p\p)--(a\a);

\node[capacity] at (3.52,3.18) {$[0,1]$};

\node[fc note] at (3.85,0.95)
  {Edge exists only if $\|p-a\|\le r$};

\foreach \j in {1,...,3}
  \draw[fc arrow,draw=fcBlue] (a\j)--(t);

\node[capacity] at (6.72,3.18)
  {$[\ell_{a,i},u_{a,i}]$};

\draw[fc arrow,densely dashed]
  (t.south)--(8.30,0.55)
  --(0.35,0.55)--(s.south);

\node[capacity] at (4.325,0.55) {};

\node[fc note] at (4.325,0.08)
  {An integral feasible flow gives an actual assignment. Repeat the same network independently for every color.};

\end{tikzpicture}%
}
\caption{Lower-bounded feasible flow for one color $i$.
Every client sends exactly one unit of flow to an eligible center.
The bounds on the center-to-sink edges ensure that the number of
color-$i$ clients assigned to each center lies in its selected
bucket interval. The same construction is applied independently
to every color.}
\label{fig:lower-bounded-flow}
\end{figure}

\noindent \textbf{Subexponential Time Bound.}  The main saving in the running time comes from the separator. A direct approach would guess all $k$ centers from the grid and would require roughly $(k/\epsilon)^{\mathcal{O}(k)}$ choices. In our recursion, however, only the centers whose balls cross the separator need to be guessed and kept active. Their total number along a recursive branch is
$\mathcal{O}_d\left(\frac{k^{1-1/d}}{\varepsilon}\right).$
The recursion has $\mathcal{O}_d(\log k)$ levels, while the separator choices, rounded color counts, and flow computations contribute only additional parameter-dependent or polynomial factors. Therefore, for fixed dimension and number of colors, the overall running time is
$n^{\mathcal{O}(d)}\left(\frac{k}{\epsilon}\right)^{ \mathcal{O}_{c,d,\lambda}(k^{1-1/d}/\epsilon)}$.
Since $k^{1-1/d}=o(k)$, the dependence on $k$ is subexponential rather than exponential. A detailed analysis of the running time is provided in Section~\ref{sec:runtime}.

\section{Algorithm} \label{algo}
We first describe the radius guessing and grid preprocessing steps that produce a finite set of candidate centers.
We then present the separator recursion, color-count bucketing, and flow-based assignment subroutine, followed by the correctness proof and running-time analysis.

\noindent \textbf{Guess 1:}\label{step:radius-enumeration} In starting of our algorithm, we want to guess the radius of the optimal clustering $C$. Suppose $R^*$ is the optimal radius. In an optimal clustering $C$, consider the cluster whose minimum enclosing ball has radius $R^*$. Based on property \ref{minpoints}, we can say, $R^*$ is the radius of the minimum enclosing ball determined by some subset of at most $d+1$ input points. Hence we enumerate all possible on all such subset and compute their minimum enclosing ball radius. This gives at most $\mathcal{O}(n^{d+1})$ many candidate radius. For a fixed dimension, this can be done in poly(n) time. So, our guess $R$ tries all possible values in increasing order.

The value $R=0$ is handled separately: clients at each distinct location must form a fair cluster, after merging all clusters at that location. We return these clusters if there are at most $k$ locations and all of them satisfy the exact fairness bounds; otherwise we continue with the positive candidate radii. This avoids using a grid of spacing zero.
 \medskip 
 
\noindent\textbf{Construction of the Graph $G_R$:} After fixing a candidate guess $R>0$, we construct a graph $G_R$ from the given set of points $P$ where each vertex of $G_R$ corresponds a point of $P$ and we put an edge between two vertex $p,q$ if $d(p,q)=\|p-q\|_2 \leq 2R$ in $\mathbb{R}^d$.
Hence come our two observations-

\begin{observation}
    Every cluster of every radius $R$-solution lies in a single connected component of $G_R$.
\end{observation}
\begin{proof}
   Consider any two points $p,q$ from a same cluster $X$ whose center is say $a$. Then, according to the construction of our graph $G_R$, $\mathrm{dist}(p,q)=\|p-q\|_2 \leq \|p-a\|_2+\|a-q\|_2 \leq 2R$. Which implies that every cluster $X$ induces a clique in $G_R$. 
\end{proof}

\begin{observation}
    If $G_R$ has more than $k$ connected components, then no radius $R$ solution clustering of at most $k$ centers exists.
\end{observation}
Hence we discard this choice and move on to the next guess.

\begin{observation}
   Suppose a connected component $H$ of $G_R$ is covered by $m (\leq k)$ radius $R$ clusters. Then, every side of its axis-parallel bounding box is $\mathcal{O}(kR)$. 
\end{observation}

\begin{proof}
Let the $m$ centers contanied in $H$ is $a_1,a_2,\dots,a_m$. Now suppose there is an edge $pq \in E(G_R)$ such that, center of $p$ is $a_i$ and center of $q$ is $a_j$ for some $i \neq j \in [m]$. By triangle inequality,  $\|a_i-a_j\| \leq \|a_i-p\|_2+\|p-q\|_2+\|q-a_j\|_2 \leq 4R$. Since, $H$ has $m$ such centers, any two centers can be connected through a path of at most $m-1$ such edges. Therefore their distance can be at most $4R(m-1)$. Finally, for any two arbitrary clients $x,y$ the distance can be at most $(4m-2)R$. So, the diameter of the component can be at most $(4m-2)R \leq (4k-2)R$. Length of each side of the bounding box is $\mathcal{O}(kR)$.
\end{proof}
If any component violates this diameter bound, reject the guess $R$.
\medskip

\noindent \textbf{Griding the Whole Instance:} For each component $H$ of $G_R$, consider its bounding box enlarged by $R$ in every direction. As every center to which a point of $H$ can be assigned will lie within distance $R$ of some points of $H$. We therefore place a grid of spacing $\delta= \frac{\epsilon R}{10\sqrt{d}}$ only inside these bounded regions and take the resulting grid points as candidate centers. We denote $F_H$ is the set of grid points for the component $H$ and $F=\bigcup_H F_H$.

All boxes use the same grid $\delta\mathbb{Z}^d$. We extend each enlarged box outward to the nearest grid hyperplanes and take all grid points in this extended box, removing duplicate locations in $F$. Thus, all distinct candidate centers are $\delta$-separated, even if boxes overlap, and the nearest grid point to every relevant radius-$R$ center is included.

\begin{lemma}\label{lem:grid-size}
For every fixed $R$, the total number of grid points are $|F|= \mathcal{O}\left(\frac{k^{d+1}d^{d/2}}{\epsilon^d}\right)$.
\end{lemma}

\begin{proof}
   Side length of each bounding box is $\mathcal{O}(kR)$. So in one coordinate direction, the number of grid cell is $\mathcal{O}\left(\frac{k \sqrt{d}}{\epsilon}\right)$. There are total at most $k$ components over $d$ coordinates. Hence, total number of grid cells are, $\mathcal{O}\left(k\left(\frac{k \sqrt{d}}{\epsilon}\right)^d\right)=\mathcal{O}\left(\frac{k^{d+1}d^{d/2}}{\epsilon^d}\right)$.   
\end{proof}

Also, notice that, any radius-$R$ center serving a nonempty cluster are distance at most $\delta \sqrt{d}/2$ from a grid point of $F$. Define, $r=R+\frac{\delta \sqrt{d}}{2}=R+\frac{\epsilon R}{20} < (1+\epsilon)R$.

\begin{lemma}\label{lem:grid-snapping}
    If an exactly $(\alpha, \beta)$-fair radius $R$ solution exists with $k$ centers, then there is an exactly $(\alpha, \beta)$-fair solution of radius at most $r$ whose centers all lie in $F$ and which uses at most $k$ centers.
\end{lemma}

\begin{proof}
    Suppose there is an exactly $(\alpha,\beta)$-fair solution of optimal radius $R$ of $k$ center clustering. Consider an arbitrary center $a$ of this solution and map it to the nearest grid point $f \in F$. Clearly, the distance between $a$ and $f$ is at most $\delta \sqrt{d}/2$. Now suppose a client $p$ is assigned to $a$ in the actual solution. Then, after moving the center to $f$, $||p-f|| \leq ||p-a||+||a-f|| \leq r.$ Hence, considering all grid points as the candidate center can increase the radius at most $r \leq (1+\epsilon)R$. If several centers map to the same grid point, merge their clusters. Property~\ref{unifair} preserves exact fairness, and the number of centers cannot increase.
\end{proof}

Now, for each $f \in F$, consider a closed ball $D_f=B(f,r)$ of radius $r$. It consists of all clients that may be assigned to $f$ without exceeding radius $r$.

\begin{definition}
    A family of balls is called $p$-ply if every point of that space can belong to at most $p$ balls.
    Formally, for a family of balls $\{D_f:f \in F\}$, its ply is $p=\max_{x\in\mathbb{R}^d}|\{f\in F:x\in D_f\}|$.
\end{definition}

\begin{lemma}\label{bddply}
    The family $\{D_f:f \in F\}$ has ply $p=\mathcal{O}((\sqrt{d}/\epsilon)^d)$.
\end{lemma}

\begin{proof}
   Consider an arbitrary point $x \in \mathbb{R}^d$. We want to count total number of balls $D_f=B(f,r)$ that can contain $x$. If $x\in D_f$, then by definition, $\|f-x\|\leq r$. So, every point $f$ whose closed ball contains $x$ must lie inside $B(x,r)$. Hence, it is sufficient to count then umber of grid points that lie in $B(x,r)$. Each grid point is $\delta$-separated. So, if we draw a ball $B(f, \delta/2)$ of radius $\delta/2$, then their interiors are disjoint. Also, each $f$ lies within distance $r$ of $x$, so all these small balls also lie inside $B(x,\delta/2+r)$, with a slightly larger radius.  Now suppose there ar $N$ such candidate centers which contains $x$. Then,  $N\cdot \operatorname{vol}\left(B\left(0,\frac{\delta}{2}\right)\right)
\le
\operatorname{vol}\left(B\left(0,r+\frac{\delta}{2}\right)\right).$ Since the volume of a $d$-dimensional ball of radius $t$ is proportional to $t^d$, $
N
\le
\left(
\frac{r+\delta/2}{\delta/2}
\right)^d
=
\left(
\frac{2r}{\delta}+1
\right)^d.$ Now $ \delta=\frac{\varepsilon R}{10\sqrt{d}}.$
Therefore, $
\frac{r}{\delta}
=
\mathcal{O}\left(\frac{\sqrt{d}}{\varepsilon}\right).$
Hence,
$N
=
\mathcal{O}\left(
\left(\frac{\sqrt{d}}{\varepsilon}+1\right)^d
\right).$ Thus, every point $x$ belongs to at most $\mathcal{O}_d(\epsilon^{-d}) $ candidate balls.
\end{proof}
\noindent \textbf{Existence of Sphere Separator:}
Here, we can use the sphere separator theorem for the family of closed balls with a bounded ply \cite{DBLP:journals/jacm/MillerTTV97}.
\begin{theorem}[Sphere separator]\label{sepexist}
    Let $\Gamma$ be a $p$-ply family of a $m$ balls in $\mathbb{R}^d$. There exists a $(d-1)$-dimensional sphere $\gamma$ such that, 
    \begin{itemize}
        \item At most $\mathcal{O}(C_d p^{1/d}m^{1-1/d})$ balls that intersects $\gamma$, where $C_d$ is a constant dependent on dimention $d$.
        \item At most $\beta_d m$ balls lie completely inside $\gamma$, where $\beta_d=\frac{d+1}{d+2} <1$.
        \item At most $\beta_d m$ balls lie completely outside $\gamma$.
    \end{itemize}
\end{theorem}

Hence by lemma \ref{bddply} and theorem \ref{sepexist} at each step applied to any set of at most $\ell$ selected candidate balls, there is a separator with $w_d(\ell)=\min\{\ell,\lceil C_d p^{1/d}\ell^{1-1/d}\rceil\}=\mathcal{O}_d(\ell^{1-1/d}/\varepsilon)$ selected balls which crosses the separator.

\noindent \textbf{Finding the Separator:} So, from the previous theorem, we know the existence of such a separator, but if we want to find such sphere separator, it is not possible to enumerate all such infinitely many possibilities. Hemce, we want to reduce the range of enumeration into finitely many choices. Consider a sphere separator in $\mathbb{R}^d$ specified as $(x_1,x_2,\dots,x_d, \rho)$, where $\rho \geq 0$ is the radius of that sphere separator. Now, a fix candidate ball $D_f=B(f,r)$ and le $h=||x-f||_2$. Notice that, the candidate ball $D_f$ has only three status relative to the separator :

1. $\rho-r>0$ and $h^2 <(\rho-r)^2$ iff the ball $D_f$ is completely inside the separator.

2. $h^2 > (\rho+r)^2$ iff $D_f$ lies completely outside the separator.

3. Neither of the above holds iff the ball is crossing or touching the separator.

Hence, for each $f$ we introduce three constant degree polynomials, 
$$g_f^-(x,\rho)=\|x-f\|^2_2-(\rho-r)^2,$$
$$g_f^+(x,\rho)=\|x-f\|^2_2-(\rho+r)^2,$$
$$h_f(x,\rho)=\rho-r.$$

For each candidate ball $D_f$, the signs of $g_f^-,g_f^+$ and $h_f$ completely determine whether $D_f$ lies inside the separator, outside or intersects it.  Since there are $M=|F|$ many candidate balls, we obtain total $3M+1$ polynomials in the $d+1$ parameters $(x_1,x_2,\dots,x_d,\rho)$. We also include the polynomial $\rho$ to restrict to nonnegative radii. Each of these polynomials has constant degree. All these polynomials divide the space into $M^{\mathcal{O}(d)}$ regions in a way that, inside one range, the sign of every polynomial remains unchanged. Hence, every separator sphere is treated as equivalent, as it gives exactly the same inside, outside, or crossing classification of all the candidate balls.

We therefore choose one representative point from every non-empty region with $\rho \geq 0$. This includes sign conditions in which some polynomials are zero, so separators touching candidate balls are not lost. Standard algorithms for realizable sign conditions compute these sample points in $M^{\mathcal{O}(d)}$ arithmetic operations for fixed $d$ \cite[Chapter~13]{basu2006algorithms}. The selected sphere represents all separator spheres in that region. 

Hence, we can sum up the whole concept and claim that,
\begin{claim}\label{clm:representative-separators}
    A list of $M^{\mathcal{O}(d)}$ representative spheres can be computed such that, for every sphere $\gamma$, the list contains a sphere separator $\gamma^*$ with the same inside, outside, and crossing classification for every candidate ball.
\end{claim}

\noindent \textbf{Solving Three Parts Recursively:} In the previous section, we showed that a separator sphere divides the selected candidate balls into three groups. So, the balls lying completely inside and completely outside can be handled separately. So, we can solve these two cases through recursion. However, a selected ball that crosses the separator may have assigned clients from both sides. So, the cluster formed by this center can not be completed immediately and has to remain active during the solution of all its subproblems. 
After choosing a separator, we define a recursive state as $sol(P', F ', A, l)$ where,
\begin{itemize}
    \item $P'=$ the set of clients that belong to the current region handled by this current recursive call. For the root, i.e., the 1st recursive call, $P=P'$. 
    \item $F'=$ all the candidate centers whose selection is not yet decided and closed balls of each center lies completely within that current region. At the start, $F=F'.$ 
    \item $A=$ The set of centers that have already been selected but whose clusters are not yet complete. We call them active centers.
    \item $l=$ The number of additional centers that may still be selected from $F'$.
\end{itemize}

For a fixed call, the table is indexed only by the bucket vector of the inherited active centers in $A$. Each entry keeps one integral witness: its selected centers, the assignment of all clients in $P'$, and the exact color counts of that assignment. Every nonempty cluster at a newly selected center is complete within this region and must pass the relaxed fairness test. Inherited active centers are not tested yet, since their contributions from other regions are still missing. If several witnesses have the same bucket vector, keep any one of them. The clients and candidate centers in a child call are always restricted to the current region. For every active center $a \in A$, the current recursive call also keeps track of how many clients of color $i$ are assigned to $a$ from the current region. We write this as a vector for each center $a$, $x_a=(x_{a,1},x_{a,2},\dots,x_{a,c})$, where, $x_{a,i}$ denotes the number of color $i$ clients assigned to $a$ from this region. When different recursive parts are later combined, their contribution to the same active center are added together. But if a state wants to keep the exact counts and enumerate on all possible values for $x_{a, i}$, it will result in a lot of states. So, to keep the number of states small, the algorithm stores a rounded version of the counts.
\medskip

\noindent \textbf{Rounding the assigned points count:} For an active center $a\in A$, the DP will store the number of clients of each color that have already been assigned to $a$. These numbers can be any integers between $0$ and $n$. To decrease the number of DP states, we group the nearby counts into multiplicative buckets and store only the corresponding bucket indices. An active center may remain active through several recursive levels. Hence, the count will be updated every time before the center is finally completed. We first bound how many such recursive levels are possible. Let $D_d =\left\lceil\log_{1/\beta_d} k \right\rceil + 1=\mathcal{O}_d(\log(k+1)),$ where $\beta_d = \frac{d+1}{d+2} < 1.$ At a recursive node with remaining center $\ell$, each child has at most $\beta_d \ell$. Thus, along any root-to-leaf recursion path, the number of undecided centers decreases as
$k,\ \beta_d k,\ \beta_d^2 k,\ldots$ until the value become at most $1$. Therefore, $D_d$ is a safe upper bound on the recursion depth.

Now we choose the width of the buckets. Let $\theta=\lambda^{1/(4(D_d+2))} > 1.$ The value of $\theta$ is chosen very close to $1$ as the same count may be rounded at different recursive levels. One rounding step will introduce at most a factor-$\theta$ change. There are at most $D_d+2$ relevant distortion steps along a recursive path, and adding contributions from different paths does not increase the error factor. Hence the total error can be at most  $\theta^{D_d+2}=\lambda^{1/4}.$

For a positive integer $x$, define $\operatorname{bucket}(x)=\left\lfloor\log_\theta x\right\rfloor.$ If $b=\operatorname{bucket}(x)$, then $\theta^b \le x < \theta^{b+1}.$ Hence, bucket $b$ represents the integer interval $I_b=\left\{z\in\mathbb{Z}_{>0} :\theta^b \le z < \theta^{b+1}\right\}.$ A special zero bucket represents the value $0$. This is necessary because multiplicative rounding cannot represent zero, and a zero color count must remain exactly zero. Since every color count is at most $n$, the number of possible buckets for one center-color pair is at most $L=2+\left\lfloor \log_\theta n \right\rfloor.$ By the definition of $\theta$, $\log_\theta n=\frac{\log n}{\log\theta}=\frac{4(D_d+2)\log n}{\log\lambda}.$ Therefore,
$L=\mathcal{O}\left(1+ \frac{D_d\log n}{\log\lambda} \right).$ Buckets whose intervals contain no integer between $1$ and $n$ are simply ignored. The following simple observation shows why this rounding is useful.

\begin{observation}\label{obs:thetaerror}
    If $x$ and $y$ are positive integers in the same bucket, then $\frac{x}{\theta} < y < \theta x.$
\end{observation}
\begin{proof}
    If $x,y>0$ and are in the same bucket $I_b$, then $\theta^b \le x,y < \theta^{b+1}.$ Since, $x < \theta^{b+1}$ and $y \geq \theta^b$ we get $x < \theta y.$ Thus, $y > \frac{x}{\theta}.$ Similarly, by exchanging the roles of $x$ and $y$, we have $y < \theta x.$ Therefore, $\frac{x}{\theta} < y < \theta x.$ Hence, replacing one witness count by another count from the same bucket changes the count by less than a multiplicative factor of $\theta$.
\end{proof}
 
\begin{observation} \label{obs:sumerror}
Let $z_1,z_2,\dots,z_t \geq 0$ and $z_1',z_2',\dots,z_t' \geq 0$. If for all $j\in[t]$ and for some $\eta \geq 1$ $\frac{z_j}{\eta} \leq z'_j \leq \eta z_j$, then $$\sum_j \frac{z_j}{\eta}\leq\sum_j z'_j\leq\eta \sum_j z_j.$$
\end{observation}

\begin{lemma}\label{lem:fairness-preservation-rounding}
    Suppose a nonempty cluster is exactly $(\alpha,\beta)$-fair in a reference assignment. Let $x_1,x_2,\dots,x_c$ be its exact color counts, and let $x_1',x_2',\dots,x_c'$ be the actual counts of another integral cluster. Assume that $x_i/\eta\leq x_i'\leq\eta x_i$ for every $i\in[c]$, where $\eta=\theta^{D_d+2}\leq\lambda^{1/4}$. These bounds preserve zero counts and imply that the new cluster is nonempty. Consequently, $$\frac{\alpha_i}{\lambda }\leq \frac{x_i'}{\sum_j x_j'} \leq \min\{1, \lambda \beta_i\} \text{ for every color } i \in [c].$$ 
\end{lemma}
\begin{proof}
    Set $\mu=\eta=\theta^{D_d+2}\leq\lambda^{1/4}$. The count-distortion bound is an assumption of this lemma; Lemma~\ref{recur-step} proves it for the witness following the reference assignment. By Property~\ref{sqmiu}, changing a factor-$\mu$ in all color counts changes every color proportion by at most a factor of $\mu^2$. Since $\mu^2\leq\lambda^{1/2}\leq\lambda,$ an original proportion in $[\alpha_i,\beta_i]$ becomes a proportion in $\left[\frac{\alpha_i}{\lambda},\,\lambda\beta_i\right].$ Since a proportion is at most $1$, the upper bound is $\min\{1,\lambda\beta_i\}.$ Hence the result follows.
\end{proof}

\noindent \textbf{Assigning the points through flow concept:} Now we show how the algorithm obtains an actual assignment of clients once the bucket intervals are fixed. We do not solve the whole fair-assignment problem directly by flow. Instead, the DP first fixes a bucket for every center-color pair, and the flow only checks whether these bucket ranges can be realized by an assignment. Let $Y\subseteq P$ be the set of clients that must be assigned to a fixed set of selected centers $A'$. For every $a\in A'$ and color $i$, suppose a bucket $b_{a, i}$ has been fixed. Let $\ell_{a,i}=\min I_{b_{a,i}}$ and $u_{a,i}=\max\bigl(I_{b_{a,i}}\cap\{0,1,\ldots,n\}\bigr).$ Thus, center $a$ must receive between $\ell_{a,i}$ and $u_{a,i}$ clients of color $i$. For the zero bucket, both values are $0$. For every color $i$, construct the following flow network:
\begin{itemize}
    \item add an edge from the source $s$ to every client $p\in Y\cap P_i$ with capacity $[1,1]$;
    \item add an edge from a client $p$ to a center $a\in A'$ whenever $\|p-a\|_2\le r,$ with capacity $[0,1]$; add an edge from every center $a$ to the sink $t$ with lower and upper capacities  $[\ell_{a,i},u_{a,i}].$
\end{itemize}

We add an edge from $t$ to $s$ with lower and upper capacities both equal to $|Y\cap P_i|$, and ask for a feasible circulation. Subtracting the lower bounds and adding a super-source and a super-sink reduces this to a standard maximum-flow feasibility test \cite{DBLP:books/daglib/0069809}. No costs are needed. The first type of edge ensures that every client is assigned exactly once. The second type allows a client to be assigned only to a center within distance $r$. The last type ensures that the number of color-$i$ clients assigned to each center lies in its fixed bucket interval. Since every client belongs to exactly one color class and there is no capacity shared between different colors at a center, the $c$ color classes can be handled by $c$ independent flow computations.

\begin{lemma}\label{lem:flow-feasibility-completeness}
    For fixed bucket intervals, feasibility of assigning $Y$ to $A'$ can be decided by $c$ independent polynomial-time flow computations. If a feasible assignment exists, an integral assignment can be obtained. Moreover, if there is a valid integral assignment of $Y$ to $A'$, then one of the bucket vectors enumerated by the algorithm gives a feasible flow instance.
\end{lemma}

\begin{proof}
For a fixed color $i$, every client sends exactly one unit of flow and can send it only to a center within distance $r$. The lower and upper capacities on the center-to-sink edges ensure that each center receives a number of color-$i$ clients lying in its fixed bucket interval. The color networks are independent because the color classes form a partition of the clients and there is no shared capacity between different colors. All capacities and lower bounds are integral. Hence, whenever the flow instance is feasible, an integral feasible flow exists, which directly gives an assignment of the clients. Now suppose a valid integral assignment already exists. For every center $a$ and color $i$, consider the exact number of color-$i$ clients assigned to $a$, and choose the bucket containing this number. Since the algorithm enumerates all bucket vectors, it eventually considers this choice. The given assignment itself satisfies all the corresponding lower and upper bounds, so the resulting flow instance is feasible. 
\end{proof}

\noindent \textbf{How one recursive call works :} After the previous steps, we now have all the information needed for the recursion. We next explain how one recursive step splits the current problem into inside and outside parts and handles the centers that cross the separator. Consider a call, $\mathrm{sol}(P',F',A,\ell) \text{ with } \ell >1.$
For every fix representative separator sphere $\gamma$, let $X_{\gamma}=\{f \in F': D_f \text{ crosses } \gamma\}.$
The algorithm guesses which of these candidate centers (crossings $\gamma$) are actually selected. It therefore again enumerates every set $Q \subset X_\gamma$ with $|Q| \leq w_d(\ell)$.
Hence, the candidate centers in $Q$ are the selected centers, while all the candidate centers in $X_\gamma \setminus Q$ are rejected in this branch. We then set $A'=A \cup Q$. 
The centers in $A'$ are active because they may receive client points from both sides of the separator. Now come the non-intersecting parts- let $F_{in}$ and $F_{out}$ be the non-crossing candidate centers whose balls lie completely inside and completely outside of $\gamma$. Similarly, $P_{in}, P_{out} $, and $P_{cross}$ are the client points that lie inside, outside, and on the separator, respectively. The algorithm then guesses the possible number of centers $\ell_{in}, \ell_{out}$ from both parts such that $\ell_{in},\ell_{out} \leq \beta_d \ell \text{ and } |Q|+\ell_{in}+\ell_{out} \leq \ell.$ The sets $P_{in},P_{out},P_{cross}$ are a partition of $P'$, and both $F_{in}$ and $F_{out}$ are subsets of $F'$. A representative separator need not place clients on the same sides as the original separator: preserving the candidate-ball classifications is sufficient, since every client served by a noncrossing ball still lies on the same side as that ball. All other assigned clients can use the selected crossing centers.

Here, it recursively solve $\mathrm{sol}(P_{in},F_{in},A',\ell_{in}) \text{ and } sol(P_{out},F_{out},A',\ell_{out})$. After both recursive calls return, their contribution to every center are combined. The boundary clients in $P_{cross}$ are then assigned to the active centers using the bucketing concept and finding a feasible flow as described in the previous steps. A ball lying completely inside $\gamma$ can not contain an outside or boundary client point. A ball completely lying outside $\gamma$ can not contain an inside or boundary client point. Consequently, after the selected crossing set $Q$ is guessed, every interaction between two recursive sides goes through only active centers in $A'=A \cup Q.$

When combining states, add their exact witness counts, rather than bucket endpoints. A compatible pair consists of two witnesses for these disjoint client sets and the same active-center set $A'$; the inside and outside counts need not be equal. For each boundary bucket vector, one integral feasible flow is sufficient. Check every nonempty newly completed cluster using its actual counts and delete unused newly completed centers. Keep every inherited center in $A$, including centers with zero contribution in the current region. Then form the bucket vector of the inherited centers and keep one complete witness for each vector. Thus, compression changes which assignment is retained; it never changes the counts of a stored assignment. A center can be active for two reasons. It may have become active at the current separator, or it may already have been active from an earlier separator. Based on this, we make the following two observations.

\begin{observation}
    For every center $q \in Q$, after combining the inside, outside, and boundary assignments, the whole cluster of $q$ is known. Therefore, its fairness can be checked, and $q$ can be removed from the active set.
\end{observation}

\begin{observation}
 Let $a\in A$ be a center that became active at an earlier separator. At the current recursive step, we may see only part of the clients that will finally be assigned to $a$. Some other clients assigned to $a$ may lie in another region created by that earlier separator. Therefore, the whole cluster of $a$ is not known yet, so $a$ must remain active.   
\end{observation}

\begin{algorithm}[H]
\caption{Group-Fair Euclidean $k$-Center in $\mathbb{R}^d$}
\label{alg:outer}
\begin{algorithmic}[1]

\STATE If $P=\emptyset$, return the empty clustering. Set $k\gets\min\{k,n\}$.

\STATE Construct the candidate-radius set $\mathcal{R}$ from all nonempty subsets of $P$ of size at most $d+1$

\FOR{each $R\in\mathcal{R}$ in increasing order}

    \IF{$R=0$}
        \STATE Test the zero-radius case described above; return if feasible, otherwise continue.
    \ENDIF

    \STATE Construct the graph $G_R$ on $P$, where
    $pq\in E(G_R)$ if $\|p-q\|_2\le 2R$.

    \IF{$G_R$ has more than $k$ connected components}
        \STATE \textbf{continue}.
    \ENDIF

    \IF{some connected component $H$ satisfies
    $\operatorname{diam}(H)>(4k-2)R$}
        \STATE \textbf{continue}.
    \ENDIF

    \STATE For every component, construct its enlarged bounding box.

    \STATE Set
    $\displaystyle
    \delta=\frac{\varepsilon R}{10\sqrt d}
    $
    and use the common grid $\delta\mathbb{Z}^d$ in the padded enlarged boxes.

    \STATE Let $F$ be the resulting set of candidate centers.

    \STATE Set
    $\displaystyle
    r=R+\frac{\delta\sqrt d}{2}.
    $

    \FOR{each $f\in F$}
        \STATE Define the candidate ball $D_f=B(f,r)$.
    \ENDFOR

    \STATE Compute the finite family of relevant separator-sphere types.

    \STATE $\mathcal{T}\gets
    Sol(P,F,\emptyset,k)$.

    \IF{$\mathcal{T}$ contains a feasible state with no active centers}
        \STATE Recover its integral clustering.
        \STATE \textbf{return} the clustering.
    \ENDIF

\ENDFOR

\STATE \textbf{return} infeasible.

\end{algorithmic}
\end{algorithm}

\begin{algorithm}[H]
\caption{$Sol (P', F', A,\ell)$}
\label{alg:solve}
\small 
\begin{algorithmic}[1]

\IF{$\ell\le 1$}
    \STATE Execute the base case and \textbf{return} the resulting table.
\ENDIF

\STATE Initialize an empty table $\mathcal{T}$.

\FOR{each relevant separator sphere $\gamma$}

    \STATE Let $X_\gamma
    =
    \{f\in F' : D_f \text{ crosses } \gamma\}.$

    \FOR{each $Q\subseteq X_\gamma$ with $|Q|\le {\min\{\ell,w_d(\ell)\}}$}

        \STATE Set $A'\gets A\cup Q$.

        \STATE Declare all centers in $X_\gamma\setminus Q$ unselected.

        \STATE Partition the clients into
        $P_{\mathrm{in}},P_{\mathrm{out}},P_{\partial}$.

        \STATE Let $F_{\mathrm{in}}$ and $F_{\mathrm{out}}$
        be the noncrossing candidate centers lying completely
        inside and outside $\gamma$, respectively.

        \FOR{each pair $(\ell_{\mathrm{in}},\ell_{\mathrm{out}})$ satisfying
        $
        \ell_{\mathrm{in}},\ell_{\mathrm{out}}
        \le \beta_d\ell,
        \text{ and }
        |Q|+\ell_{\mathrm{in}}+\ell_{\mathrm{out}}\le \ell
        $}

            \STATE
            $\mathcal{T}_{\mathrm{in}}
            \gets
            \textsc{Solve}
            (P_{\mathrm{in}},F_{\mathrm{in}},A',
            \ell_{\mathrm{in}})$.

            \STATE
            $\mathcal{T}_{\mathrm{out}}
            \gets
            \textsc{Solve}
            (P_{\mathrm{out}},F_{\mathrm{out}},A',
            \ell_{\mathrm{out}})$.

            \FOR{each compatible pair of states
            $(\tau_{\mathrm{in}},\tau_{\mathrm{out}})$ from
            $\mathcal{T}_{\mathrm{in}}$ and
            $\mathcal{T}_{\mathrm{out}}$}

                \STATE Combine the contributions of the two states
                for all centers in $A'$.
                \FOR{each rounded bucket vector for the boundary contributions}
                \color{black}

                \STATE Use the bucket-constrained flow to assign
                $P_{\partial}$ to centers in $A'$.

                \IF{the flow instance is infeasible}
                    \STATE \textbf{continue}.
                \ENDIF

                \STATE Add the inside, outside, and boundary
                contributions to obtain the exact counts for
                every center in $A'$.

                \STATE Set $\mathrm{valid}\gets\mathrm{true}$.

                \FOR{each $q\in Q$}
                    \STATE Let $X_q$ be its completed witness cluster; if $X_q=\emptyset$, omit $q$ from this witness and continue.

                    \IF{for some $i\in[c]$,
                    $
                    \frac{|X_q\cap P_i|}{|X_q|}
                    \notin
                    \left[
                    \frac{\alpha_i}{\lambda},
                    \min\{1,\lambda\beta_i\}
                    \right]
                    $}
                        \STATE $\mathrm{valid}\gets\mathrm{false}$.
                    \ENDIF

                \ENDFOR

                \IF{$\mathrm{valid}=\mathrm{false}$}
                    \STATE \textbf{continue}.
                \ENDIF

                \STATE Forget all newly completed centers in $Q$.

                \STATE Round the accumulated color counts of the
                inherited active centers in $A$.

                \STATE Insert the resulting state into $\mathcal{T}$,
                keeping one integral witness for each rounded state.

                \ENDFOR
            \ENDFOR
        \ENDFOR
    \ENDFOR
\ENDFOR

\STATE \textbf{return} $\mathcal{T}$.

\end{algorithmic}
\end{algorithm}

\section{Correctness} \label{correct}
After describing the algorithm, we now prove its correctness. We first handle the base case, where at most one new center can still be chosen, and then prove the recursive case for $\ell>1$. The main idea is to show that if a correct fair solution exists, then at least one branch of the algorithm follows it and is never discarded.

\noindent \textbf{Base Case: when $\ell \leq 1$ .} Consider a recursive call $Sol(P',F',A,\ell)$ with $\ell \leq 1.$

If $\ell=0$, no new center can be selected. Hence, all clients in $P'$ must be assigned to the already active centers in $A$.

If $\ell=1$, the algorithm considers both possibilities: either no new center is selected, or exactly one center $f\in F'$ is selected. In the latter case, every possible choice of $f$ is taken.

For each such choice, let $A''$ be the set of centers available for the assignment. Thus, $A''=A$
if no new center is chosen, and $A''=A\cup\{f\}$ if $f$ is selected. The algorithm then enumerates the bucket values for the color counts contributed by the clients in $P'$ and uses the flow construction to check whether these clients can be assigned to the centers in $A''$. For every feasible flow, test the actual completed cluster of the new center $f$, if it is nonempty, against the relaxed fairness bounds. Discard a witness failing this test, and delete $f$ from this witness if it is unused. Inherited centers in $A$ remain active. Insert each surviving witness under the bucket vector of its exact contributions to $A$, keeping one witness per vector. If both the client set and the available center set are empty, the empty assignment is feasible; if clients remain but no center is available, it is infeasible.
\begin{lemma} \label{base-step}
    Suppose the current base-case call has a radius-$r$ reference assignment using at most $\ell$ new centers, whose complete clusters at new centers are exactly $(\alpha,\beta)$-fair. For inherited active centers, only their contributions from $P'$ are prescribed. Then the algorithm finds a feasible state corresponding to this solution, up to the allowed rounding on the number of color counts.
\end{lemma}

\begin{proof}
Choose the same zero-or-one new center as in the reference solution. For every center and color, place its exact count into the corresponding bucket. By Lemma \ref{lem:flow-feasibility-completeness}, the flow for these bucket choices is feasible. Any integral assignment returned for these same buckets has each color count within a factor $\theta$ of the corresponding reference count, while a zero count remains zero. By property \ref{sqmiu}, each color proportion therefore changes by at most a factor $\theta^2\le \lambda$. Hence every newly completed center satisfies the relaxed fairness condition. After table compression, the retained contribution to every inherited center differs from the reference contribution by at most $\theta^2$: one factor for choosing a flow witness and one for keeping a representative with the same returned bucket vector. Zeros remain zero. No completed cluster is changed by this last operation; every retained witness has already passed its fairness test.
\end{proof}

\noindent \textbf{Recursive Step: when $\ell >1$ }

\begin{lemma} \label{recur-step}
{Suppose a call $Sol(P',F',A,\ell)$ has a radius-$r$ reference assignment using at most $\ell>1$ additional centers from $F'$, with complete exactly fair clusters at these new centers. For inherited active centers, only the regional color contributions are prescribed. If this call has remaining recursion height $h$ (a base case has height zero), its returned table contains an integral witness whose contribution to every inherited center-color pair is within a factor $\theta^{h+2}$ of the reference contribution, with zeros preserved. Every completed cluster of this witness is $\lambda$-relaxed fair.}
\end{lemma}

\begin{proof}
We prove the stated invariant by induction on the remaining recursion height. The base case follows from Lemma~\ref{base-step} and its proof.

For $\ell>1$, let $S\subseteq F'$ be the new centers used by the reference assignment, after deleting unused centers. If $S=\emptyset$, choose any representative sphere and take $Q=\emptyset$ and both child budgets zero. Otherwise, apply Theorem~\ref{sepexist} to the $|S|\leq\ell$ selected balls. Claim~\ref{clm:representative-separators} supplies an enumerated sphere with the same classifications. Choose $Q$ to be precisely the crossing centers of $S$, and choose the child budgets to be the numbers of selected noncrossing centers on the respective sides. These guesses satisfy $|Q|\leq w_d(\ell)$, both child budgets are at most $\beta_d\ell$, and their sum with $|Q|$ is at most $\ell$.

Restrict the reference assignment to the client sets defined by this representative sphere. Every noncrossing selected ball serves clients only on its own side. Thus, each child reference assignment completes exactly fair clusters at its own new centers, while centers in $A\cup Q$ remain inherited active centers. The induction hypothesis gives child witnesses whose active counts differ from these reference contributions by at most $\theta^{h+1}$. Enumerate the buckets of the reference boundary contributions. Lemma~\ref{lem:flow-feasibility-completeness} supplies an integral boundary witness with distortion at most $\theta$. By Observation~\ref{obs:sumerror}, adding the three contributions has distortion at most $\theta^{h+1}$, with zeros preserved. For $q\in Q$, this sum is the complete cluster. The reference cluster is nonempty and exactly fair, so the witness cluster is also nonempty and passes the relaxed fairness test by Lemma~\ref{lem:fairness-preservation-rounding}. All clusters completed within the children already passed this test and are kept unchanged. For inherited centers in $A$, keeping one witness with the same returned bucket vector introduces at most one further factor $\theta$, by Observation~\ref{obs:thetaerror}. This proves the bound $\theta^{h+2}$. Even if many witnesses are inserted under this vector, the final witness and the reference-following candidate lie in the same bucket, so only one such factor is needed. Since $h\leq D_d$, the total count distortion is at most $\theta^{D_d+2}$, as used in the fairness lemma.
\end{proof}

Now comes our main theorem- 

\begin{theorem}
Fix $d\geq 2$, $c\geq 1$, $0<\varepsilon\leq 1$, and $\lambda>1$. Let $P\subseteq\mathbb{R}^d$ be partitioned into $c$ color classes with fairness bounds $(\alpha_i,\beta_i)$. If there exists an exactly $(\alpha,\beta)$-fair clustering using at most $k$ continuous centers, then the algorithm returns a clustering using at most $k$ centers whose radius is at
most $(1+\epsilon)R^\star,$
where $R^\star$ is the optimum radius of an exactly fair solution. Moreover, every nonempty returned cluster \(X\) satisfies $$\frac{\alpha_i}{\lambda} \leq \frac{|X\cap P_i|}{|X|}
\leq \min\{1,\lambda\beta_i\}\text{ for every }i\in[c].$$  
\end{theorem}

\begin{proof}
The optimum radius \(R^\star\) is among the radii considered by the algorithm. If $R^\star=0$, the separate zero-radius test returns an exactly fair solution. Hence assume $R^\star>0$. For this radius, Lemma~\ref{lem:grid-snapping} gives an exactly fair solution using candidate centers from $F$, with radius $r<(1+\epsilon) R^\star$. Now consider the root call $Sol(P, F,\emptyset,k).$ Lemma \ref{base-step} shows that a correct branch is preserved in every base case. For a recursive call with $\ell>1$, Lemma \ref{recur-step} shows that the algorithm considers a separator, a set of selected crossing centers, and inside and outside budgets that agree with the exact fair solution. Therefore, by induction over the recursion, at least one feasible branch survives up to the root. Every client in this branch is assigned exactly once, either in one of the recursive subproblems or through the boundary flow. An assignment is allowed only if the client is within distance $r$ of its assigned center. The center-budget conditions also ensure that at most $k$ centers are selected.

Whenever the complete cluster of a center becomes known, its fairness is checked. These tests use the actual integral witness counts. Later compression only chooses between already valid witnesses and never modifies a completed cluster, so every returned witness is sound. For the witness supplied by the completeness proof, Lemma~\ref{lem:fairness-preservation-rounding} ensures that these tests succeed. Hence every nonempty completed cluster satisfies
$$\frac{\alpha_i}{\lambda} \leq \frac{|X\cap P_i|}{|X|} \leq
\min\{1,\lambda\beta_i\}\text{ for every }i\in[c].$$ At the root, there are no remaining active centers, so the surviving branch gives a complete clustering. It uses at most $k$ centers, has radius at most $r<(1+\varepsilon)R^\star,$ and satisfies the required relaxed fairness condition for every cluster. If the algorithm returns earlier at a candidate radius $R_0<R^\star$, the same soundness argument applies and its radius is at most $(1+\varepsilon)R_0\leq(1+\varepsilon)R^\star$.
\end{proof}

\section{Runtime Analysis}\label{sec:runtime}
We first recall the two bounds needed for the running-time analysis. Since each recursive child has center budget at most \(\beta_d\ell\), where $\beta_d=\frac{d+1}{d+2}<1,$ the recursion depth is

$
D_d=
\left\lceil
\log_{1/\beta_d}k
\right\rceil+1
=
\mathcal{O}_d(\log(k+1)).
$ Moreover, at a level with budget \(\ell\), at most $
\mathcal{O}_d\left(\frac{\ell^{1-1/d}}{\varepsilon}\right)
$ new crossing centers become active. Summing this along one root-to-leaf path gives $H_d= \mathcal{O}_d\left(
\frac{k^{1-1/d}}{\varepsilon}
\right).$ Thus, the recursion has depth \(\mathcal{O}_d(\log k)\), and at any time at most \(H_d\) centers need to be kept active. We now count the different contributions to the running time.

\noindent \textbf{Radius Enumeration.}
By the radius-enumeration step in Section~4, the algorithm considers $\mathcal{O}(n^{d+1})$
candidate radii. For fixed $d$, this is $n^{\mathcal{O}(d)}$.
\medskip

\noindent \textbf{Candidate Centers.}
By Lemma~\ref{lem:grid-size}, the number of grid candidate centers is $$M
=
\mathcal{O}\left(
\frac{k^{d+1}d^{d/2}}{\varepsilon^d}
\right)
=
\left(\frac{k}{\varepsilon}\right)^{\mathcal{O}(d)} \text{ for fixed } d.$$
\medskip

\noindent \textbf{Separator Types.}
Claim~\ref{clm:representative-separators} gives $M^{\mathcal{O}(d)}$ representative separator types at one recursive node. Since the recursion depth is at most $D_d$, their contribution along a recursive path is $M^{O_d(D_d)}.$
\medskip

\noindent \textbf{Crossing-Center Guesses.}
At a node with budget $\ell$, at most $w_d(\ell)
=
\mathcal{O}_d\left(
\frac{\ell^{1-1/d}}{\varepsilon}
\right)$
selected balls cross the separator. Hence the number of possible choices for the selected crossing set $Q$ is at most
$\sum_{j=0}^{w_d(\ell)}
\binom{M}{j}
\le
M^{\mathcal{O}_d(\ell^{1-1/d}/\varepsilon)}.$
Along one root-to-leaf path, these exponents sum to $\mathcal{O}_d(H_d)$.
Therefore, the total contribution of the crossing-set guesses is
$M^{\mathcal{O}_d(H_d)}.$ The two recursive children are solved separately, so their running times are added. The cost of combining their DP tables is included in the state-table bound below.

For clarity, let $B(\ell)$ be the number of separator, crossing-set, and budget guesses at a node. Then $B(\ell)\leq M^{\mathcal{O}(d)}(k+1)^2 M^{\mathcal{O}_d(\ell^{1-1/d}/\varepsilon)}$. If $W$ is a uniform bound on local table processing, the recursive cost satisfies $T(\ell)\leq B(\ell)(2T(\lfloor\beta_d\ell\rfloor)+W)$ for $\ell>1$, with the base-case cost included in $W$. Unrolling this bound counts all calls by the product of these branching factors along a path, with an additional $2^{\mathcal{O}(D_d)}$ factor. Local table processing is performed once per call; its state factor is not multiplied once per recursion level.

\medskip
\noindent \textbf{Budget guesses.}
At each recursive level, there are at most $\mathcal{O}(k^2)$
choices for $(\ell_{\mathrm{in}},\ell_{\mathrm{out}})$. Over depth $D_d$,
this contributes $k^{\mathcal{O}(D_d)}.$

\medskip
\noindent \textbf{Rounded DP states.}
Each active center has $c$ color-count coordinates, and each coordinate has at most $L$ possible buckets. Since at most $H_d+1$ centers are active, the number of rounded count states is at most $L^{c(H_d+1)}.$
Combining child tables and enumerating the bucket vectors used by the flow step changes only the constant in the exponent. Hence the table-processing cost is $L^{\mathcal{O}(cH_d)}k^{\mathcal{O}(1)}.$
Each individual flow computation takes polynomial time. The base case considers at most $M+1$ choices for its zero-or-one new center; this additional factor is absorbed into the power of $M$ below.

Combining all these factors gives $T_{c,d}
=
n^{\mathcal{O}(d)}
M^{\mathcal{O}_d(H_d+D_d)}
L^{\mathcal{O}(cH_d)}
k^{\mathcal{O}(D_d)},$
where
$$
H_d
=
\mathcal{O}_d\left(
\frac{k^{1-1/d}}{\epsilon}
\right),
\text{ }
D_d=\mathcal{O}_d(\log(k+1)), \text{ and }
L
=
\mathcal{O}\left(
1+
\frac{D_d\log n}{\log\lambda}
\right).
$$

The factor involving $\log n$ can be absorbed into a polynomial factor in $n$ together with a parameter-dependent factor.

Indeed, for $a\geq1$ and $n\geq1$, $(1+\log n)^a\leq n(a+1)^a$. Taking $a=\mathcal{O}(c(H_d+1))$ and using $L=\mathcal{O}((1+D_d/\log\lambda)(1+\log n))$ gives $L^a\leq n[\mathcal{O}((1+D_d/\log\lambda)(a+1))]^a$. The exponent of $n$ remains independent of $k$. Therefore, for fixed
$d$,
$$
T_{c,d}
=
n^{\mathcal{O}(d)}
\left(\frac{k}{\varepsilon}\right)^{
\mathcal{O}_d(k^{1-1/d}/\varepsilon)
}
\left(
\frac{c\,k^{1-1/d}}{\varepsilon}
\left(1+\frac{\log(k+1)}{\log\lambda}\right)
\right)^{
\mathcal{O}_d(c\,k^{1-1/d}/\varepsilon)
}.
$$
In particular, $T_{c,d}
=
n^{\mathcal{O}(d)}
\left(\frac{k}{\varepsilon}\right)^{
\mathcal{O}_{c,d,\lambda}(k^{1-1/d}/\varepsilon)
}.$

\section{Conclusion}\label{sec:conclusion}
We studied the group -fair $k$-center problem in fixed dimension, where the input points are divided into disjoint classes, and every cluster must satisfy given fairness conditions. By combining grid-based discretization, a sphere-separator method, and a recursive algorithm, we obtain a $(1+\epsilon)$- approximation for the radius while allowing a multiplicative $\lambda$ -relaxation of the fairness constraints. The main idea is to use a separator to divide the problem into smaller parts. During the recursion, we keep only a small number of active centers and store their color counts using multiplicative buckets. This reduces the number of recursion states while still preserving fairness within the allowed relaxation. 

Compared with earlier work on fair clustering in more general metric spaces, we studied the problem on Euclidean space to obtain better dependence on $k$. As a direction for further study, it would be interesting to remove the fairness relaxation, handle overlapping color groups, and study whether a similar subexponential approach can be extended to more general geometric settings.

\section{AI-Assisted Preparation}
The initial research ideas and mathematical intuition originated with the authors. Through interactive discussions, OpenAI's ChatGPT 6.1 Sol assisted in formalizing these ideas, including the application of the geometric sphere-separator lemma to the proposed algorithmic framework, addressing gaps in the arguments, and exploring possible algorithmic improvements. The manuscript was composed by the authors, with AI assistance in refining selected passages. All mathematical arguments and algorithmic claims were independently verified by the authors, who take full responsibility for the final content.

\bibliographystyle{alpha}
\bibliography{arxiv}


\end{document}